\documentclass[letterpaper]{article} %
\usepackage{configuration/aaai2026}  %
\usepackage{times}  %
\usepackage{helvet}  %
\usepackage{courier}  %
\usepackage[hyphens]{url}  %
\usepackage{graphicx} %
\usepackage{natbib}  %
\usepackage{caption} %
\usepackage{cuted}
\usepackage{algorithm}
\usepackage{algorithmic}

\usepackage{newfloat}
\usepackage{listings}
\DeclareCaptionStyle{ruled}{labelfont=normalfont,labelsep=colon,strut=off} %
\floatstyle{ruled}
\newfloat{listing}{tb}{lst}{}
\floatname{listing}{Listing}
\newcommand{\lib}{\textsc{DP-GenG}\ }
\newcommand{\libn}{\textsc{DP-GenG}}
\title{\lib: Differentially Private Dataset Distillation \\ Guided by DP-Generated Data}

\author {
  Shuo Shi\textsuperscript{\rm 1},
  Jinghuai Zhang\textsuperscript{\rm 2},
  Shijie Jiang\textsuperscript{\rm 3},
  Chunyi Zhou\textsuperscript{\rm 1},\\
  Yuyuan Li\textsuperscript{\rm 4},
  Mengying Zhu\textsuperscript{\rm 1},
  Yangyang Wu\textsuperscript{\rm 1},
  Tianyu Du\textsuperscript{\rm 1}\thanks{Corresponding Author}
}
\affiliations {
  \textsuperscript{\rm 1}Zhejiang University,
  \textsuperscript{\rm 2}University of California, Los Angeles,
  \textsuperscript{\rm 3}Hangzhou Normal University,
  \textsuperscript{\rm 4}Hangzhou Dianzi University\\
  \{shuoshi, zjradty\}@zju.edu.cn
}

\usepackage{amsmath,amssymb,amsthm}

\providecommand{\abs}[1]{\left\lvert#1\right\rvert}
\providecommand{\norm}[1]{\left\lVert#1\right\rVert}

\providecommand{\cF}{\mathcal{F}}

\providecommand{\cM}{\mathcal{M}}
\providecommand{\cN}{\mathcal{N}}
\providecommand{\cO}{\mathcal{O}}

\providecommand{\cS}{\mathcal{S}}
\providecommand{\cT}{\mathcal{T}}

\newenvironment{talign*}
{\csname align*\endcsname}
{\endalign}

\usepackage[utf8]{inputenc}         %
\usepackage[T1]{fontenc}            %
\usepackage{url}                    %
\usepackage{booktabs}               %
\usepackage{amsfonts}               %
\usepackage{nicefrac}               %
\usepackage{microtype}              %
\usepackage{xcolor}                 %
\usepackage{algorithm}
\usepackage{algorithmic}
\usepackage{graphicx}
\usepackage{subcaption}
\usepackage[flushleft]{threeparttable}
\usepackage{float}
\usepackage{multirow}
\usepackage{xspace}
\usepackage{natbib}
\usepackage{enumitem}
\usepackage[font=small]{caption}
\usepackage{autobreak}
\usepackage{sidecap}
\usepackage{bbding}
\usepackage[toc, page, header]{appendix}
\usepackage{tikz}
\usepackage{xcolor}
\usepackage{pifont}
\usepackage{mdframed}
\usepackage{colortbl}
\usepackage{mathrsfs}

\definecolor{coral}{RGB}{255,127,80}
\definecolor{darkgreen}{RGB}{0,100,0}
\definecolor{darkyellow}{RGB}{204,153,0}
\definecolor{salmon}{RGB}{250,128,114}

\newcommand{\thmref}[1]{Theorem~\ref*{#1}}

\newcommand{\lemref}[1]{Lemma~\ref*{#1}}

\newcommand{\figref}[1]{Figure~\ref*{#1}}

\newcommand{\tabref}[1]{Table~\ref*{#1}}

\newcommand{\secref}[1]{Section~\ref*{#1}}

\newcommand{\appref}[1]{Appendix~\ref*{#1}}
\newcommand{\algref}[1]{Alg.~\ref*{#1}}

\newtheoremstyle{custom}
{1pt} %
{1pt} %
{\itshape} %
{} %
{\bfseries} %
{} %
{ } %
{\thmname{#1} \thmnumber{#2} \thmnote{(#3)} . } %

\theoremstyle{custom}

\newtheorem{innerdefinition}{Definition}
\newtheorem{innerproposition}{Proposition}
\newtheorem{innerassumption}{Assumption}
\newtheorem{innerremark}{Remark}
\newtheorem{innertheorem}{Theorem}
\newtheorem{innerhypothesis}{Hypothesis}
\newtheorem{innerconjecture}{Conjecture}
\newtheorem{innerlemma}{Lemma}
\newtheorem{innercorollary}{Corollary}
\newtheorem{innerexample}{Example}
\newtheorem{innernotation}{Notation}
\newtheorem{innerclaim}{Claim}
\newtheorem{innerproblem}{Problem}

\newtheorem{innerobservation}{Observation}

\newmdenv[
  backgroundcolor=white!10,
  linecolor=black!100,
  linewidth=0.8pt,
  skipabove=2pt,
  skipbelow=2pt,
  innertopmargin=5pt,
  innerbottommargin=5pt,
  innerleftmargin=5pt,
  innerrightmargin=5pt,
]{definitionframe}

\newmdenv[
  backgroundcolor=white!10,
  linecolor=black!100,
  linewidth=0.8pt,
  skipabove=2pt,
  skipbelow=2pt,
  innertopmargin=5pt,
  innerbottommargin=5pt,
  innerleftmargin=5pt,
  innerrightmargin=5pt,
]{propositionframe}

\newmdenv[
  backgroundcolor=green!10,
  linecolor=green!100,
  linewidth=0.8pt,
  skipabove=2pt,
  skipbelow=2pt,
  innertopmargin=10pt,
  innerbottommargin=5pt,
  innerleftmargin=5pt,
  innerrightmargin=5pt,
]{assumptionframe}

\newmdenv[
  backgroundcolor=white!10,
  linecolor=black!100,
  linewidth=0.8pt,
  skipabove=2pt,
  skipbelow=2pt,
  innertopmargin=5pt,
  innerbottommargin=5pt,
  innerleftmargin=5pt,
  innerrightmargin=5pt,
]{remarkframe}

\newmdenv[
  backgroundcolor=white!10,
  linecolor=black!100,
  linewidth=0.8pt,
  skipabove=2pt,
  skipbelow=2pt,
  innertopmargin=5pt,
  innerbottommargin=5pt,
  innerleftmargin=5pt,
  innerrightmargin=5pt,
]{theoremframe}

\newmdenv[
  backgroundcolor=purple!10,
  linecolor=purple!100,
  linewidth=0.8pt,
  skipabove=2pt,
  skipbelow=2pt,
  innertopmargin=10pt,
  innerbottommargin=5pt,
  innerleftmargin=5pt,
  innerrightmargin=5pt,
]{hypothesisframe}

\newmdenv[
  backgroundcolor=orange!10,
  linecolor=orange!100,
  linewidth=0.8pt,
  skipabove=2pt,
  skipbelow=2pt,
  innertopmargin=10pt,
  innerbottommargin=5pt,
  innerleftmargin=5pt,
  innerrightmargin=5pt,
]{conjectureframe}

\newmdenv[
  backgroundcolor=white!10,
  linecolor=black!100,
  linewidth=0.8pt,
  skipabove=2pt
  skipbelow=2pt,
  innertopmargin=5pt,
  innerbottommargin=5pt,
  innerleftmargin=5pt,
  innerrightmargin=5pt,
]{lemmaframe}

\newmdenv[
  backgroundcolor=magenta!10,
  linecolor=magenta!100,
  linewidth=0.8pt,
  skipabove=2pt,
  skipbelow=2pt,
  innertopmargin=10pt,
  innerbottommargin=5pt,
  innerleftmargin=5pt,
  innerrightmargin=5pt,
]{corollaryframe}

\newmdenv[
  backgroundcolor=lime!10,
  linecolor=lime!100,
  linewidth=0.8pt,
  skipabove=2pt,
  skipbelow=2pt,
  innertopmargin=10pt,
  innerbottommargin=5pt,
  innerleftmargin=5pt,
  innerrightmargin=5pt,
]{exampleframe}

\newmdenv[
  backgroundcolor=pink!10,
  linecolor=pink!100,
  linewidth=0.8pt,
  skipabove=2pt,
  skipbelow=2pt,
  innertopmargin=10pt,
  innerbottommargin=5pt,
  innerleftmargin=5pt,
  innerrightmargin=5pt,
]{notationframe}

\newmdenv[
  backgroundcolor=violet!10,
  linecolor=violet!100,
  linewidth=0.8pt,
  skipabove=2pt,
  skipbelow=2pt,
  innertopmargin=10pt,
  innerbottommargin=5pt,
  innerleftmargin=5pt,
  innerrightmargin=5pt,
]{claimframe}

\newmdenv[
  backgroundcolor=salmon!10,
  linecolor=salmon!100,
  linewidth=0.8pt,
  skipabove=2pt,
  skipbelow=2pt,
  innertopmargin=10pt,
  innerbottommargin=5pt,
  innerleftmargin=5pt,
  innerrightmargin=5pt,
]{problemframe}

\newmdenv[
  backgroundcolor=lavender!10,
  linecolor=lavender!100,
  linewidth=0.8pt,
  skipabove=2pt,
  skipbelow=2pt,
  innertopmargin=10pt,
  innerbottommargin=5pt,
  innerleftmargin=5pt,
  innerrightmargin=5pt,
]{observationframe}

\newenvironment{definition}
{\begin{definitionframe}\begin{innerdefinition}}
      {\end{innerdefinition}\end{definitionframe}}

\newenvironment{theorem}
{\begin{theoremframe}\begin{innertheorem}}
      {\end{innertheorem}\end{theoremframe}}

\newenvironment{lemma}
{\begin{lemmaframe}\begin{innerlemma}}
      {\end{innerlemma}\end{lemmaframe}}

\usepackage[textwidth=3.5cm,textsize=tiny]{todonotes}

\newcommand{\IPC}{\texttt{IPC}\xspace}



\begin{document}

\maketitle

\begin{abstract}
  Dataset distillation (DD) compresses large datasets into smaller ones while preserving the performance of models trained on them. Although DD is often assumed to enhance data privacy by aggregating over individual examples, recent studies reveal that standard DD can still leak sensitive information from the original dataset due to the lack of formal privacy guarantees. Existing differentially private (DP)-DD methods attempt to mitigate this risk by injecting noise into the distillation process. However, they often fail to fully leverage the original dataset, resulting in degraded realism and utility. This paper introduces \libn, a novel framework that addresses the key limitations of current DP-DD by leveraging DP-generated data.
  Specifically, \lib initializes the distilled dataset with DP-generated data to enhance realism. Then, generated data refines the DP-feature matching technique to distill the original dataset under a small privacy budget, and trains an expert model to align the distilled examples with their class distribution. Furthermore, we design a privacy budget allocation strategy to determine budget consumption across DP components and provide a theoretical analysis of the overall privacy guarantees.
  Extensive experiments show that \lib significantly outperforms state-of-the-art DP-DD methods in terms of both dataset utility and robustness against membership inference attacks, establishing a new paradigm for privacy-preserving dataset distillation.
  Our code is available at \url{https://github.com/shuoshiss/DP-GENG}.
\end{abstract}

\section{Introduction}
\label{sec:intro}

\begin{figure}[!h]
  \centering
  \includegraphics[width=0.5\textwidth]{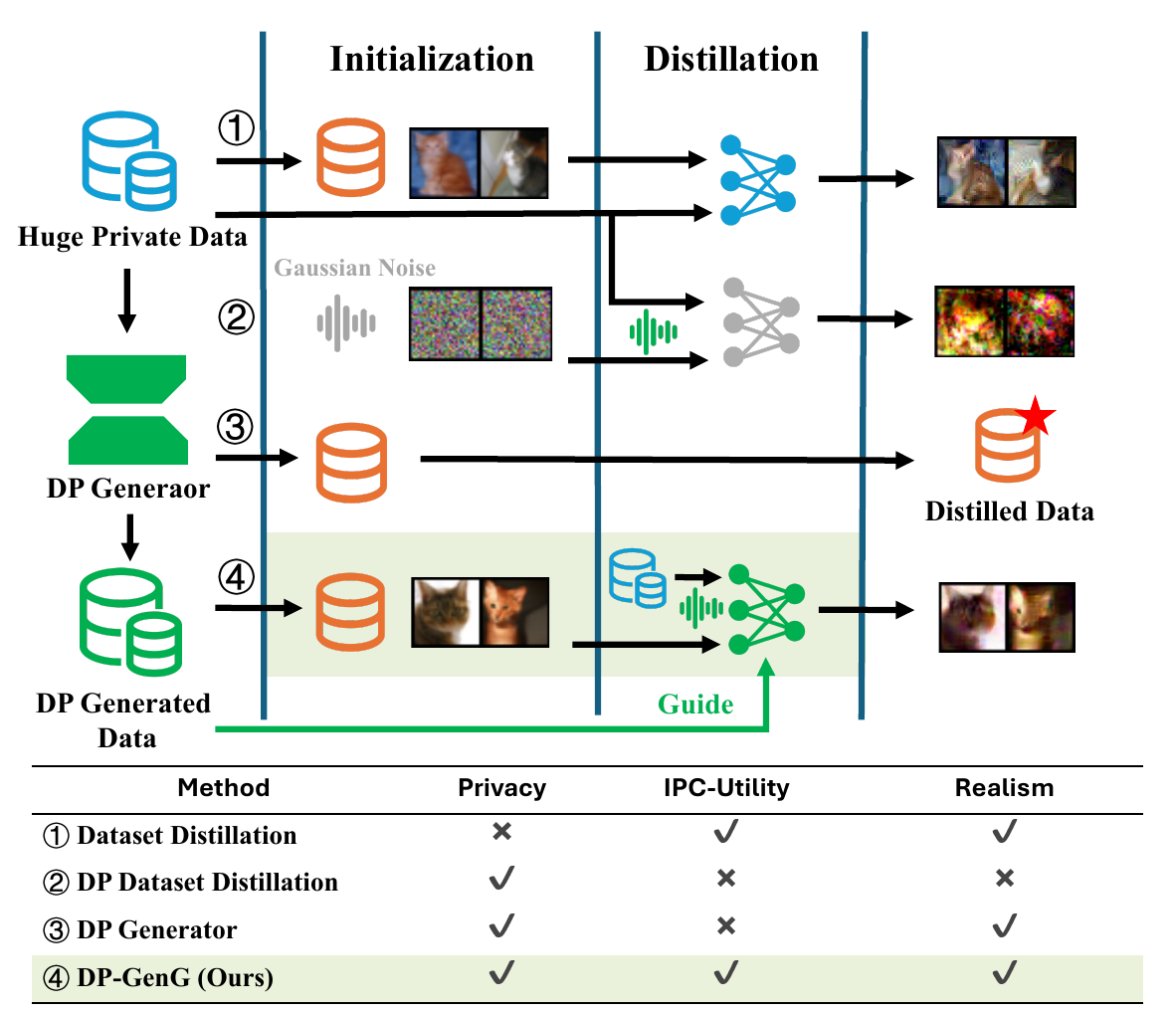}
  \caption{We compare different dataset distillation methods in terms of the privacy, utility and realism of the resulting distilled datasets.}
  \label{fig:dp_distillation_concept}
\end{figure}

Dataset distillation (DD)~\citep{wang2018dataset} aims to compress large datasets into smaller ones while preserving the utility of models trained on them. This technique offers compelling advantages: it enhances training efficiency, reduces storage demands~\citep{yu2023dataset}, and can improve data privacy when working with sensitive data~\citep{dong2022privacy, tong2025dapt, tong2024mmdfnd}. Recent advances~\citep{guo2024lossless,sun2024diversity,wang2025ncfm} have shown that models trained on distilled datasets can achieve accuracy comparable to those trained on full datasets-even when using orders of magnitude fewer training examples.
\looseness=-1

Despite the promising performance of DD, a critical privacy issue remains: although the distilled dataset is assumed to inherently protect data privacy due to its compressed nature, recent studies~\citep{carlini2022no,chen2022private} have shown that standard distillation techniques offer no formal privacy guarantees and can still leak sensitive information. This limitation raises significant concerns, particularly when applied to sensitive datasets—such as medical records—where even marginal leakage of private details could violate privacy regulations (e.g., GDPR~\citep{regulation2016regulation}) or expose individuals to harm~\citep{Karale2021The, lu2025dammfnd, lan2025contextual}.

To address this issue, existing methods (e.g., PSG~\cite{chen2022private} and NDPDC~\cite{zheng2024differentially}) propose differentially private dataset distillation (DP-DD) techniques. The key idea is to inject Gaussian noise into the distillation process, allowing the creation of a distilled dataset from the private dataset while ensuring formal differential privacy (DP) guarantees. While these methods limit an adversary's ability to infer the presence of individual data points, they suffer from key limitations that result in suboptimal performance, as shown in \figref{fig:dp_distillation_concept}. \textbf{Limitation 1:} They fail to preserve the realism of each distilled example due to the lack of direct access to natural data. Realism quantifies how visually and semantically consistent an example is with its corresponding label. A low realism score - often stemming from poor initialization or noisy optimization signals - can severely compromise the utility of the distilled dataset~\citep{sun2024diversity,shao2024elucidating}. \textbf{Limitation 2:} Given a limited privacy budget, these methods must add large amounts of noise, which further degrades both the quality and utility of the resulting dataset.

Previous methods primarily enforce DP by injecting noise directly into the distillation process, which limits their ability to fully exploit the information in the private dataset. Motivated by recent advances in DP-synthetic data generation, which produces synthetic data that closely resemble the private data distribution while ensuring formal DP guarantees, we propose \lib (\underline{D}ifferentially \underline{P}rivate \underline{Gen}erator-\underline{G}uided Distillation). \lib is a novel DP-DD pipeline that bridges DP-generated data with DD to address key limitations in existing DP-DD frameworks.
Specifically, this work aims to address the following research questions: \textbf{(RQ1)} \textit{How can we leverage DP-generated data to guide the distillation process and address key limitations?} \textbf{(RQ2)} \textit{How can we increase the utility of distilled datasets under limited privacy budgets?}
\looseness=-1

\lib begins by using DP data generators to produce a large volume of synthetic data that closely resembles the original private dataset. Then, it selects a representative subset of DP-generated data to initialize the distillation process, which substantially promotes the realism of distilled examples (\textbf{solution to L1}). In the subsequent phase, \lib introduces a novel DP feature-matching technique, which aligns the small distilled dataset with the original dataset in the feature space. Unlike previous methods~\citep{zhao2023idm,zhang2023accelerating,zhang2024dance} that train multiple feature extractors directly on the private dataset for meaningful guidance, \lib constructs feature extractors using the DP-generated data. Due to the post-processing property of DP~\citep{dwork2014algorithmic}, this design incurs no additional privacy cost, thereby mitigating the impact of budget limitations \textbf{(solution to L2)}. Moreover, we identify a critical issue: directly applying DP to DD can degrade performance due to the noises introduced in the matching process, which may cause the distilled examples to drift away from their intended class representation. To address this, \lib introduces a novel DP expert model that acts as a calibrator to dynamically align each distilled example with samples from the same class during distillation.

To further ensure high-quality distillation within a limited privacy budget, \lib designs a tailored privacy budget allocation algorithm that determines budget consumption across different DP components (e.g., synthetic data generation, feature matching, expert guidance) and provides a theoretical analysis of the overall privacy guarantees ensured by \libn. Extensive experiments demonstrate that \lib significantly outperforms existing DP-DD methods in both utility and privacy trade-offs. Our contributions can be summarized as follows:
\begin{itemize}
    \item We propose \libn, a novel framework that bridges DP-generated data with differentially private dataset distillation. \lib identifies key limitations in existing DP-DD methods and leverages tailored mechanisms, guided by DP-generated data, to address them.
    \item We design a privacy budget allocation strategy across DP components and provide a theoretical analysis of the overall privacy guarantees ensured by \libn.
    \item Our evaluations show \lib not only outperforms prior DPDD methods by 11.6\% but is also highly robust to membership inference attacks.
\end{itemize}

\section{Related Work}
\label{sec:related_work}
\noindent
\textbf{Dataset Distillation (DD)}.
DD aims to compress large datasets into smaller synthetic ones while preserving their utility for model training. Early works \citep{wang2018dataset,zhao2020dataset,zhao2023dataset} pioneer gradient matching and distribution matching approaches. Recent advances have further improved DD through more advanced techniques, such as trajectory matching \citep{cazenavette2022dataset,guo2024lossless}, and parameterization\citep{liu2022dataset,yuan2024color}.

While DD is often assumed to protect data privacy~\citep{li2020soft,dong2022privacy}, recent studies~\citep{li2024data,zhao2025does} show that standard DD methods remain vulnerable to membership inference attacks~\citep{carlini2022membership}. To address this, many works have proposed integrating DP into DD to develop DP-DD methods with formal privacy guarantees, including PSG \citep{chen2022private} and NDPDC \citep{zheng2024differentially}. However, these methods simply inject noise into the distillation process and fail to fully exploit the information in the private dataset, leading to suboptimal results.

\paragraph{Differential Privacy (DP).}
DP \citep{dwork2006calibrating} provides a theoretical framework for privacy protection by ensuring that the presence or absence of any individual record in a dataset has minimal impact on an algorithm’s output (See basic definition of DP in Appendix~\ref{definition}). This is achieved by injecting calibrated noise, which prevents adversaries from inferring sensitive information about individual data points while maintaining statistical utility \citep{dwork2006calibrating,patwa2023dp,sun2024netdpsyn,farayola2024data}. \citet{shokri2017membership} show that applying DP to a learning task can reduce the success rate of privacy attacks. Besides,
\citet{jayaraman2019evaluating} evaluate the effectiveness of $(\epsilon, \sigma)$-DP and its variants in neural networks by using membership inference attacks~\citep{shokri2017membership}.

Prior to our work, DP-DD methods are based on Rényi Differential Privacy (RDP)~\citep{mironov2017renyi}. In this work, we adopt $f$-DP~\citep{dong2022gaussian} to achieve improved performance, which offers lossless privacy accounting through a hypothesis testing framework. More specifically, consider testing $H_0$: data $\sim P$ vs. $H_1$: data $\sim Q$ with a rejection rule $\phi \in [0, 1]$. The type I error is $\alpha_\phi = \mathbb{E}_P[\phi]$, and the type II error is $\beta\phi = 1 - \mathbb{E}_Q[\phi]$. The trade-off function $T(P, Q)(\alpha)$ in $f$-DP gives the minimal type II error for type I error at level $\alpha$: $T(P, Q)(\alpha) = \inf_{\phi} \{ \beta_\phi : \alpha_\phi \leq \alpha \}$.
Specifically, this work utilizes \(\mu\)-Gaussian Differential Privacy (\(\mu\)-GDP), an instance of f-DP where the trade-off function is given by \(f=G_{\mu}\), with $G_\mu(x) = \Phi(\Phi^{-1}(1-x) - \mu)$ and $\Phi$ denoting the cumulative distribution function of the standard normal distribution (see \appref{definition} for details).

\paragraph{DP-synthetic Data Generation.}
The practical demand for privacy-preserving synthetic data drives the development of DP-synthetic data generation methods. They aim to generate a dataset (called DP-generated data) that retains the properties of the private dataset while protecting individual privacy~\citep{jia2025prada}. For example, \citet{ghalebikesabi2023differentially} fine-tunes a pre-trained generator on the private dataset with DP guarantees, then uses the fine-tuned generator to produce DP-generated data. In contrast, \citet{lin2023differentially} directly uses a pre-trained generator to create DP-generated data, injecting noise during the dataset refinement process.

\section{Methodology}
\label{sec:method}
In this part, we present our framework, \lib, which integrates DP-generated data to improve the performance of DP-DD. Following~\citep{chen2022private,zheng2024differentially}, we focus on DP-DD in the image domain.
\secref{method:data} describes the techniques used to generate synthetic datasets under DP guarantees. \secref{method:DPFM} explains how DP-generated data is used to guide each component of the DD process. \secref{method:DPEG} introduces an expert model designed to mitigate issues introduced by DP noise.
Finally, \secref{method:privacy} details the budget allocation and analyzes the overall privacy guarantees ensured by \lib.

\begin{figure*}[t]
  \centering
  \includegraphics[width=0.8\textwidth]{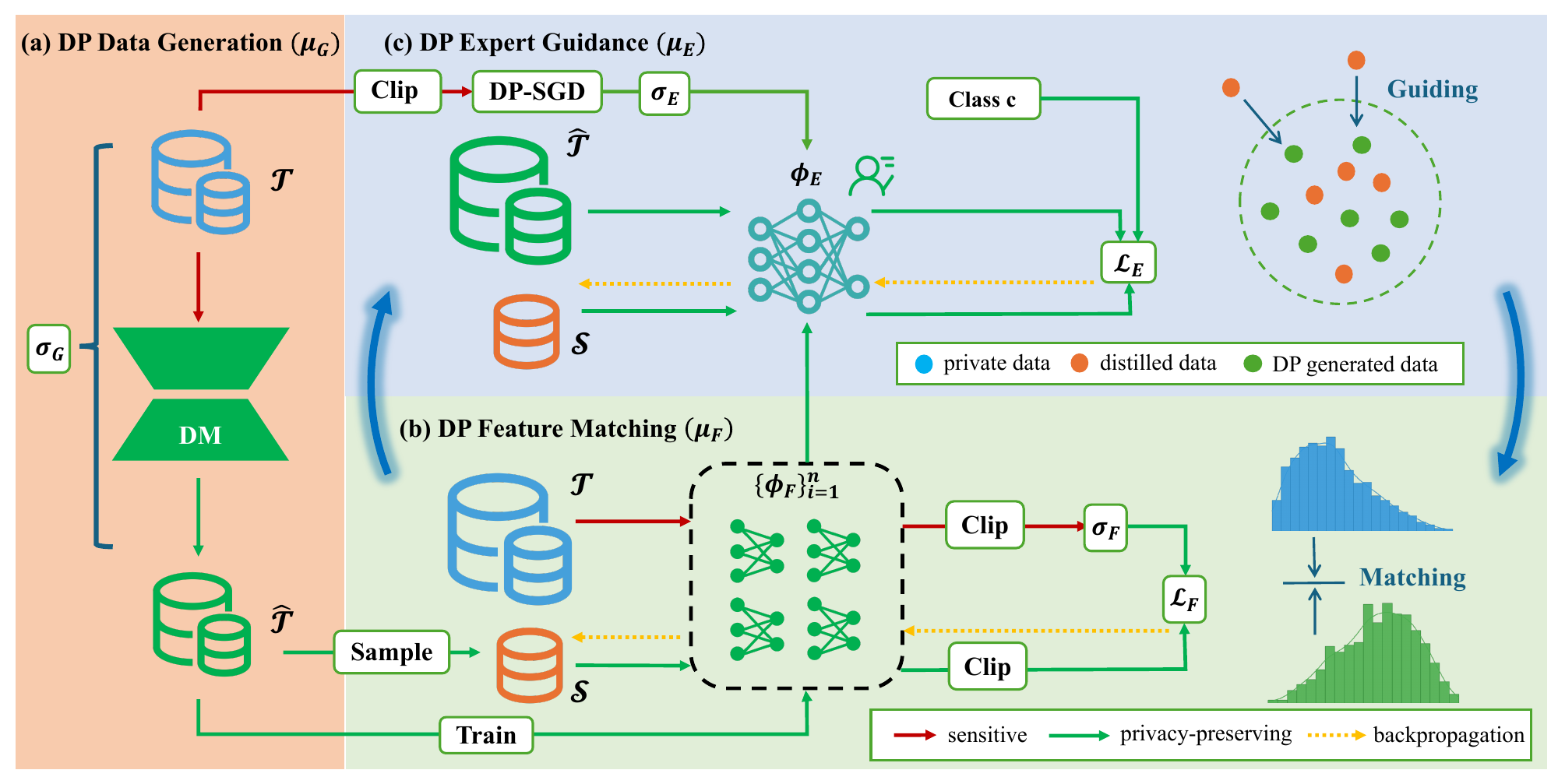}
  \caption{The overall framework of \lib. It fully leverages DP-generated data throughout the distillation process to enhance the performance of the distilled dataset. The blue, green and orange datasets represent the original private dataset, its DP-generated version and its distilled dataset (with DP guarantees), respectively.}
  \label{fig:framework}
\end{figure*}

\subsection{DP Data Generation}
\label{method:data}
Inspired by recent advances in DP image synthesis, we aim to generate a synthetic dataset $\hat{\mathcal{T}}$ that closely resembles the private dataset $\mathcal{T}$ while providing formal privacy guarantees.
Existing DP image synthesis approaches differ in how they ensure privacy and can be categorized into input-level~\citep{harder2021dp}, model-level~\citep{ghalebikesabi2023differentially,liu2023differentially} and output-level~\citep{lin2023differentially} methods, depending on the stage at which DP noise is applied. Among them, the most widely adopted methods involve pre-training generators on large public datasets and subsequently fine-tuning them on the private dataset $\mathcal{T}$ using DP mechanisms (e.g., PrivImage~\citep{li2024privimage}), or injecting noise directly during the synthesis and refinement phases (e.g., PE~\cite{lin2023differentially}). The two approaches are the primary focus of this paper.
Notably, these approaches rely on the Gaussian mechanism, where Gaussian noise $\xi$ is introduced either during the fine-tuning phase of DP data generators~ \citep{li2024privimage} or in the voting results of evolutionary algorithms~\citep{lin2023differentially}. For such methods, adding Gaussian noise with a standard deviation $\sigma_{\text{G}} = 1/\mu_{\text{G}}$ ensures a privacy guarantee of $\mu_{\text{G}}$-GDP, as formalized below:
\begin{lemma}[Gaussian Mechanism to GDP~\citep{dong2022gaussian}]
  Define the Gaussian mechanism that operates on a statistic $\theta$ as $\cM(D)=\theta(D)+\xi$, where $\xi\sim\cN(0, \mathrm{sens}(\theta)^2/\mu^2)$. Then, $\cM$ is $\mu$-GDP.
\end{lemma}
DP image synthesis approaches can produce a large number of synthetic samples $\hat{\mathcal{T}}$ that capture diverse characteristics of the original private dataset. Moreover, due to the post-processing property of DP~\citep{dwork2014algorithmic}, the DP-generated data $\hat{\mathcal{T}}$ will inherit the privacy guarantees of the synthesis process. In other words, $\hat{\mathcal{T}}$ can be freely utilized in downstream computations (distillation process), without incurring any additional privacy cost.

\begin{theorem}[Post-processing~\citep{dwork2014algorithmic}]
  \label{thm:post_processing}
  If \(\cM\) satisfies \((\epsilon, \delta)\)-DP, \(\cF\circ\cM\) will satisfy \((\epsilon, \delta)\)-DP for any private data independent function \(\cF\) with \(\circ\) denoting the composition operator.
\end{theorem}

Motivated by this, \textbf{the key idea is to construct a large synthetic dataset that can replace the private dataset during certain parts of the distillation process.} In contrast to prior methods that enforce DP solely by injecting noise during the optimization, the use of DP-generated data offers a novel perspective on achieving privacy guarantees, allowing us to fully leverage the knowledge contained in the private dataset without incurring large privacy costs while improving the realism of the distilled dataset.

\subsection{DP Feature Matching}
\label{method:DPFM}
Feature matching (FM) is an effective method for DD, where the distilled dataset is optimized to match the feature distributions of the full dataset. We adopt FM as the core distillation algorithm within DP-DD due to its generalizability. Moreover, we incorporate DP-generated data in both the initialization and matching phases to address identified limitations.

\paragraph{Initialization with DP-generated Data.}
A recent study shows that DD methods achieve optimal utility when distilled examples are initialized with samples from the private dataset~\citep{shao2024elucidating, zhao2023benchmark}, which is attributed to the enhanced realism provided by such initialization. Unlike these approaches, as well as previous DP-DD methods that rely on Gaussian noise initialization~\citep{chen2022private,zheng2024differentially}, we propose using DP-generated data to initialize the distilled dataset. Given that the size of the distilled dataset is controlled by the number of images per class (\IPC), we construct the initialized distilled set $\mathcal{S}_{\text{init}}$ by applying a sampling strategy (e.g., selection by performing k-means clustering on features) over the DP-generated data $\hat{\mathcal{T}}$: $\cS_\text{init}=\text{Sample}(\hat{\mathcal{T}}, \IPC)$. Moreover, we employ a parameterization technique~\citep{kim2022dataset}, which embeds multiple DP synthetic images within a single image, to maximize the utilization of the private dataset. Notably, the use of DP-generated data for initialization directly addresses limitation \textbf{(L1)} by preserving the realism of the distilled dataset, which in turn leads to enhanced utility-privacy tradeoff.

\paragraph{Feature Extractor from DP-generated Data.}
The feature extractor in FM captures image representations to facilitate the alignment of feature distributions.
The success of FM relies on multiple trained feature extractors~\citep{zhao2023idm,zhang2024dance, zhao2024balf}.
Existing methods either use private dataset directly to train the feature extractors, which violates DP guarantees, or rely on randomly initialized feature extractors, which results in poor utility. A feasible alternative is to obtain DP-guaranteed feature extractors using DP-SGD~\citep{abadi2016deep}; however, this approach is often impractical under limited privacy budgets.
Specifically, achieving robust feature extraction requires training multiple feature extractors $\{\phi_{\text{F}}^{n}\}_{n=1}^{N}$ in parallel over multiple iterations. Applying DP-SGD in this setting introduces substantial noise due to the limited privacy budgets allocated to each training run, which ultimately degrades the quality of the distilled dataset. In this work, we propose training feature extractors using DP-generated data, which retains the ability to capture features of the private dataset while effectively addressing the budget limitation \textbf{(L2)}.  According to \thmref{thm:post_processing}, models trained on DP-generated data inherit its privacy guarantees without incurring additional privacy cost.

\paragraph{Matching via Injected DP Noise.}
Given trained feature extractors, the matching phase aligns the feature distribution of the distilled dataset with that of a full dataset. This full dataset can be either DP-generated data or the original private dataset. The only difference is that, when using the original private dataset, DP must be enforced by injecting noise during the matching process.
Empirically, we find that using the original private dataset yields better performance under the same privacy budget, despite the added noise.
We attribute this to the fact that the distilled dataset is initialized with DP-generated data, which already captures prior knowledge from that source. As a result, the original private dataset provides a more informative guidance signal during matching. Specifically, we inject Gaussian noise with standard deviation $\sigma_F$ during the matching process, as shown below:
\begin{equation}
  \label{eq:mean_feature_def}
  \bar{\phi}_{\mathrm{F}}(\mathcal{D}) = \frac{1}{\abs{\mathcal{D}}}\sum_{\boldsymbol{x} \in \mathcal{D}}\textsc{Clip}(\phi_{\mathrm{F}}(\boldsymbol{x}), C),
\end{equation}
\begin{equation}
  \label{eq:feature_matching_loss_redefined}
  \mathcal{L}_{\text{F}} = \left\| \bar{\phi}_{\mathrm{F}}(\mathcal{T}) + \mathcal{N}(0,\sigma_\text{F}^2C^2I) - \bar{\phi}_{\mathrm{F}}(\mathcal{S}) \right\|^2.
\end{equation}
where $\textsc{Clip}(\phi_{\mathrm{F}}(\boldsymbol{x}^{\mathcal{*}}), C)=\min(1, \frac{C}{\norm{\phi_{\mathrm{F}}(\boldsymbol{x}^{\mathcal{*}})}})\phi_{\mathrm{F}}(\boldsymbol{x}^{\mathcal{*}})$ represents a feature clipping mechanism designed to regulate the sensitivity of the L2 norm of the averaged feature vector.

Unlike prior work~\cite{chen2022private,zheng2024differentially} that relies on Rényi Differential Privacy~\citep{mironov2017renyi}, we adopt GDP~\citep{dong2022gaussian}. Due to mini-batch processing,
by applying the \textbf{Subsampling Theorem} of GDP in \appref{appendix:privacy_analysis}, we achieve noise reduction proportional to the sampling probability $p$, while preserving the same privacy guarantees. The resulting privacy parameter is given by: $\mu_F=p\sqrt{T(\mathrm{e}^{1/\sigma_F^2}-1)}$.

\subsection{DP Expert Guidance}
\label{method:DPEG}

After initialization, examples in the distilled dataset largely retain their visual appearance throughout the matching process. However, their feature representations are continuously optimized to capture the characteristics of the full dataset.
Due to the effects of DP noise introduced in DP-DD algorithms, we observe that some distilled examples may undergo significant shifts in their feature representations. As illustrated in \figref{fig:appendix_dist_shift} in \appref{appendix:class_distribution_shift}, these shifts can cause an example to deviate from its original class, degrading the utility of the distilled dataset.
To address this, we introduce an expert model to regularize the matching process and keep features of distilled examples aligned with their class distributions.

\paragraph{Expert Model.}
We employ an expert model to learn class-wise feature representations, which enables the calibration of misaligned distilled examples affected by DP noise.
To train this expert, we first perform standard training on DP-generated data, followed by fine-tuning with DP-SGD on the original private dataset with $\mu_{\text{E}}$ (See \algref{alg:gradient_sanitizing} in \appref{appendix:expert_model_training}).

\paragraph{DP-Generated Data Guided Alignment.}
We aim to align distilled examples with their class representations by using the expert model as a regularizer. For each distilled example, we sample reference data points $\mathcal{R}^{\boldsymbol{y}^{\cS}}$ from its class $\boldsymbol{y}^{\cS}$ and encourage the example to follow the same feature distribution as these references~\citep{ma2024optimal}. To preserve privacy, we use samples from the DP-generated data as references without additional privacy cost. Since the expert model is trained to accurately capture class-wise distinctions, the alignment task becomes equivalent to maximizing the similarity between each distilled example and its associated references in the expert model’s label space. For each distilled example, we aggregate the predicted logits of its references and apply two forms of supervision: (1) KL divergence between the distilled example’s logits and the aggregated soft labels to guide distributional alignment, and (2) hard-label regularization to preserve class identity. We illustrate the loss at each iteration below, where $\mathcal{R}^{\boldsymbol{y}^{\cS}} \in \hat{\mathcal{T}}$ represents references that are randomly sampled from class $\boldsymbol{y}^{\cS}$:
\begin{equation}
  \label{eq:class_centroid_def}
  \boldsymbol{c}_{y} = \frac{1}{|{\mathcal{R}}^{y}|}\sum_{\boldsymbol{x}^{\mathcal{R}} \in {\mathcal{R}}^{y}}\phi_{\mathrm{E}}(\boldsymbol{x}^{\mathcal{R}}).
\end{equation}
\begin{equation}
  \label{eq:expert_guidance_loss_redefined}
  \mathcal{L}_{\text{E}} =\frac{1}{|\cS|}\sum_{(\boldsymbol{x}^{\cS},\boldsymbol{y}^{\cS}) \in \cS} \left[ \mathcal{L}_{\text{CE}}(\phi_\mathrm{E}(\boldsymbol{x}^{\cS}), \boldsymbol{y}^{\cS}) + \mathcal{L}_{\text{KL}}(\phi_{\mathrm{E}}(\boldsymbol{x}^{\mathcal{S}})\Vert \boldsymbol{c}_{\boldsymbol{y}^{\cS}})\right].
\end{equation}
\subsection{Overall Privacy Analysis}
\label{method:privacy}

\lib involves three components that jointly consume the privacy budget: DP data generation $\mu_G$, feature matching $\mu_F$, and expert model training $\mu_E$. In this part, we first analyze the total privacy cost of \libn. We then present a practical strategy for selecting optimal parameters to achieve a strong privacy-utility trade-off.
The algorithm of \lib is shown in \algref{alg:workflow} in \appref{appendix:expert_model_training}.

\paragraph{Privacy Cost Calculation for the Entire Process.}
The privacy parameters \(\mu_\mathrm{G}\), \(\mu_\mathrm{F}\), and \(\mu_\mathrm{E}\), which are controlled by the standard deviations of Gaussian noise, determine the total privacy cost.
Based on the composition property of GDP, we compute the overall privacy parameter \(\mu_\text{total}\) for the entire \lib workflow as below:

\begin{lemma}[GDP Composition~\citep{dong2022gaussian}]~
  \newline
  \label{lem:gdp_composition}
  The n-fold composition of \(\mu_i\)-GDP mechanisms is \(\sqrt{\mu_1^2+\cdots+\mu_n^2}\)-GDP.
\end{lemma}
Next, we need to convert the \(\mu\)-GDP to \((\epsilon, \delta)\)-DP. We can easily connect GDP and DP by \lemref{lem:conversion}. Under the same DP privacy budget, using GDP provides tighter bounds for the subsampled Gaussian mechanism. 
\begin{lemma}[GDP to DP Conversion~\citep{dong2022gaussian}]
  \label{lem:conversion}
  A mechanism is \(\mu\)-GDP if and only if it is \((\epsilon, \delta(\epsilon))\)-DP for all \(\epsilon \ge 0\), where
  \begin{equation*}
    \delta(\epsilon) = \Phi\left(-\frac{\epsilon}{\mu}+\frac{\mu}{2}\right)-\mathrm{e}^\epsilon\Phi\left(-\frac{\epsilon}{\mu}-\frac{\mu}{2}\right)
  \end{equation*}
  where \(\Phi\) denotes the standard normal CDF.
\end{lemma}

\paragraph{Budget Allocation Strategy.}
The overall privacy cost of \lib depends on the amount of noise injected into its components: DP data generation ($\sigma_\text{G}
$), DP feature matching ($\sigma_\text{F}$), and DP expert model training ($\sigma_\text{E}$). The corresponding GDP parameters are $\mu_\text{G}$, $\mu_\text{F}$, and $\mu_{E}$. We strategically allocate the privacy budget by prioritizing the noise levels for DP data generation and expert model training ($\sigma_\text{G}$ and $\sigma_\text{E}$), since their performance can be independently evaluated. Specifically, we employ binary search to determine the appropriate noise multiplier $\sigma_\text{G}$ for DP data generation, aiming to achieve a target FID score. Similarly, we adjust $\sigma_\text{E}$ for training the expert model to reach a desired accuracy (a detailed explanation of the target FID score and accuracy can be found in the \appref{appendix:target_FID_accuracy}). Once the noise multipliers $\sigma_\text{G}$ and $\sigma_\text{E}$ are set based on these utility goals, we compute the required noise level $\sigma_\text{F}$ for DP feature matching using~\lemref{lem:gdp_composition}, ensuring that the overall privacy budget $(\epsilon_{\text{total}}, \delta_{\text{total}})$ is satisfied. In other words, we allocate the privacy budget to each component according to a utility-driven criterion defined by the target FID of the DP-generated data and the target accuracy of the expert model. Our strategy ensures that the total privacy expenditure remains within the specified budget, as detailed in the following theorem, proofed in \appref{appendix:privacy_analysis}.

\begin{theorem}[\lib Privacy Budget Allocation]
  In the \lib, each component involves differential privacy guarantees. Given the total privacy parameters \(\epsilon_{\text{total}}\) and \(\delta_{\text{total}}\), with \(\mu_{\text{G}}\) for generation and \(\mu_{\text{E}}\) for expert guidance stages, we can derive the noise parameter \(\sigma_{\text{F}}\) for feature matching with batch sampling probability \(p=\frac{B}{n}\) over \(T\) iterations as follows:
  \begin{equation*}
    \sigma_{\text{F}} = \sqrt{\ln{\frac{T \cdot p^2}{\mu_{\text{total}}^2 - \mu_{\text{G}}^2 - \mu_{\text{E}}^2} + 1}}^{-1}.
  \end{equation*}
\end{theorem}

\begin{table*}[t]
  \centering
  \small
  \setlength\tabcolsep{1pt}

  \begin{tabular}{c|ccc|ccc|ccc}
    \toprule
    & \multicolumn{3}{c|}{CIFAR-10} & \multicolumn{3}{c|}{CIFAR-100} & \multicolumn{3}{c}{CelebA} \\
    \midrule
    \IPC & 1 & 10 & 50 & 1 & 10 & 50 & 1 & 10 & 50 \\
    \midrule
    \multicolumn{10}{c}{$\varepsilon=1$} \\
    \midrule
    DP-MERF &\(14.6_{\pm 0.7}\) &\(19.4_{\pm 0.3}\) & \(21.0_{\pm 0.4}\) &\(2.3_{\pm 0.1}\) &\(3.4_{\pm 0.0}\) & \(3.6_{\pm 0.1}\)& \(53.6_{\pm 0.4}\) &\(64.6_{\pm 0.8}\) & \(68.3_{\pm 0.3}\)\\
    PE & \(15.2_{\pm 0.2}\)& \(25.0_{\pm 1.0}\)& \(37.7_{\pm 1.0}\)& \(3.0_{\pm 0.3}\)& \(7.4_{\pm 0.3}\)& \(11.1_{\pm 0.2}\)& \(54.8_{\pm0.3}\) &\(65.1_{\pm 0.3}\) & \(70.6_{\pm 0.3}\)\\
    PrivImage & \(14.7_{\pm 0.2}\)& \(21.2_{\pm 0.3}\)& \(38.0_{\pm 0.2}\)& \(1.2_{\pm 0.1}\)& \(2.0_{\pm 0.1}\)& \(3.8_{\pm 0.1}\)& \(54.5_{\pm 1.3}\) &\(64.2_{\pm 0.7}\) &\(72.1_{\pm 0.6}\) \\
    \midrule
    PSG &\(26.9_{\pm 0.7}\) &\(33.7_{\pm 0.2}\) &\(35.9_{\pm 0.6}\) &\(4.9_{\pm 0.3}\) &\(8.3_{\pm 0.3}\) &\(10.3_{\pm 0.1}\) &\(66.8_{\pm 3.8}\) &\(72.8_{\pm 2.1}\) & \(76.2_{\pm 2.3}\)\\
    NDPDC & \(26.1_{\pm 0.4}\)&\(39.8_{\pm 0.2}\) &\(42.6_{\pm 0.6}\)&\(7.8_{\pm 0.2}\) &\(10.9_{\pm 0.1}\) &\(11.5_{\pm 0.3}\) &\(66.0_{\pm 1.7}\) &\(77.7_{\pm 0.5}\) & \(80.4_{\pm 0.6}\)\\
    \textbf{\lib (Ours)} & \(\bf{29.8_{\pm0.5}}\)&\(\bf{53.5_{\pm0.3}}\)&\(\bf{56.9_{\pm0.3}}\) &\(\bf{15.8_{\pm0.3}}\) & \(\bf{21.1_{\pm 0.4}}\)& \(\bf{25.9_{\pm 0.2}}\)& \(\bf{67.8_{\pm 0.7}}\)&\(\bf{78.7_{\pm 0.3}}\) &\(\bf{82.1_{\pm 0.1}}\)\\
    \midrule
    \multicolumn{10}{c}{$\varepsilon=10$} \\
    \midrule
    DP-MERF & \(14.9_{\pm 0.3}\) & \(21.7_{\pm 0.3}\) & \(22.9_{\pm 0.5}\)& \(2.8_{\pm 0.3}\) & \(3.4_{\pm 0.2}\) & \(3.8_{\pm 0.2}\) &\(54.1_{\pm 0.6}\) &\(68.5_{\pm 0.2}\) & \(72.2_{\pm 0.7}\)\\
    PE & \(16.0_{\pm 0.4}\)& \(29.8_{\pm 0.4}\)& \(42.1_{\pm 0.3}\)& \(3.4_{\pm 0.1}\)& \(10.0_{\pm 0.2}\)& \(15.1_{\pm 0.2}\)& \(56.1_{\pm 0.4}\)& \(67._5{\pm 0.5}\)& \(75.3_{\pm 0.4}\)\\
    PrivImage & \(15.2_{\pm 0.4}\)& \(27.7_{\pm 0.5}\)& \(39.6_{\pm 0.5}\)& \(1.9_{\pm 0.1}\)& \(2.8_{\pm 0.2}\)& \(4.3_{\pm 0.2}\)& \(55.2_{\pm 0.4}\)& \(65.7_{\pm 0.2}\)& \(76.8_{\pm 0.3}\)\\
    \midrule
    PSG &\(27.9_{\pm 0.2}\) &\(40.3_{\pm 0.4}\) &\(47.2_{\pm 0.6}\) &\(10.4_{\pm 0.2}\) &\(18.0_{\pm 0.2}\) &\(19.7_{\pm 0.3}\) &\(67.1_{\pm 1.3}\) &\(77.0_{\pm 0.5}\) &\(81.6_{\pm 0.9}\) \\
    NDPDC &\(26.6_{\pm 1.2}\) &\(46.8_{\pm 0.6}\) &\(53.9_{\pm 0.2}\) &\(10.7_{\pm 0.2}\) &\(17.5_{\pm 0.7}\) &\(19.2_{\pm 0.3}\) &\(67.0_{\pm 2.2}\) &\(78.1_{\pm 1.1}\) & \(82.3_{\pm 0.7}\)\\
    \textbf{\lib (Ours)} &\(\bf{35.6_{\pm 0.9}}\) &\(\bf{59.0_{\pm 0.3}}\) &\(\bf{65.5_{\pm 0.4}}\) &\(\bf{19.3_{\pm0.3}}\) & \(\bf{27.9_{\pm 0.4}}\)& \(\bf{32.3_{\pm 0.2}}\)& \(\bf{69.3_{\pm 0.4}}\)&\(\bf{81.4_{\pm 0.4}}\) & \(\bf{85.7_{\pm 0.3}}\)\\
    \midrule
    \multicolumn{10}{c}{without DP guarantees} \\
    \midrule
    {DM} &\(28.6_{\pm 0.6}\) &\(48.9_{\pm 0.6}\) & \(64.0_{\pm 0.4}\)& \(11.4_{\pm 0.3}\)& \(29.7_{\pm 0.3}\)& \(43.6_{\pm 0.4}\)& \(68.4_{\pm 0.6}\)&\(80.1_{\pm 0.3}\) & \(85.2_{\pm 0.4}\)\\
    {NCFM} &\(49.5_{\pm 0.3}\) & \(71.8_{\pm 0.3}\)& \(77.4_{\pm 0.3}\)& \(34.4_{\pm 0.5}\)& \(48.7_{\pm 0.3}\)& \(54.7_{\pm 0.2}\)& \(74.6_{\pm 0.4}\)& \(84.8_{\pm 0.3}\)& \(88.2_{\pm 0.4}\)\\
    {Whole Dataset} & \multicolumn{3}{c|}{\(84.8_{\pm 0.1}\)} & \multicolumn{3}{c|}{\(56.2_{\pm 0.3}\)} & \multicolumn{3}{c}{\(95.6_{\pm 0.3}\)} \\
    \bottomrule
  \end{tabular}
  \caption{\textbf{Comparison with previous methods on test accuracy (\%).} Results are averaged over three random seeds. \IPC denotes images per class, and $\varepsilon$ is privacy budget. NCFM is the SOTA non-private distillation algorithm. }
  \label{tab:res_lowresolution}
\end{table*}
\begin{table*}[t]
  \centering
  \small
  \setlength\tabcolsep{3pt}
  \begin{tabular}{l|cc|cc|cc|cc}
    \toprule
    Method                  & \multicolumn{4}{c|}{CIFAR-10} & \multicolumn{4}{c}{CIFAR-100} \\
    \midrule
    $\varepsilon$           & \multicolumn{2}{c|}{1} & \multicolumn{2}{c|}{10} & \multicolumn{2}{c|}{1} & \multicolumn{2}{c}{10} \\
    \midrule
    Metric                  & Test Acc. & TPR@0.1\%FPR & Test Acc. & TPR@0.1\%FPR & Test Acc. & TPR@0.1\%FPR & Test Acc. & TPR@0.1\%FPR \\
    \midrule
    PSG                     & 35.9      & 0.08         & 47.2      & 0.13         & 10.3      & 0.09         & 19.7      & 0.20 \\
    NDPDC                   & 42.6      & 0.10         & 53.9      & 0.14         & 11.5      & 0.11         & 19.2      & 0.18 \\
    \textbf{DP-GENG (Ours)} & 55.9      & 0.10         & 65.5      & 0.14         & 25.9      & 0.12         & 32.3      & 0.17 \\
    \bottomrule
  \end{tabular}
  \caption{Comparison of different methods against MIA with \IPC=50. For reference, a standard non-DP method (DM) achieves Test Acc./TPR@0.1\%FPR of 63.0/0.82 on CIFAR-10 and 43.6/1.06 on CIFAR-100.}
  \label{tab:defense_mia}
\end{table*}

\section{Experiment}
\label{sec:exp}
\subsection{Experiment Setup}
\noindent
\textbf{Datasets and Models.}
Unlike previous works on DP-DD~\citep{chen2022private, zheng2024differentially}, which focus on simple datasets such as MNIST~\citep{lecun1998gradient} and FashionMNIST~\citep{xiao2017fashion}, we conduct experiments on the CIFAR-10~\citep{krizhevsky2009learning}, CIFAR-100~\citep{krizhevsky2009learning}, and CelebA~\citep{liu2015deep} datasets. These datasets present significant challenges for DP-DD, especially under a limited privacy budget. Following previous studies~\citep{guo2024lossless, zheng2024differentially}, we use ConvNet~\citep{sagun2017empirical} as the default backbone architecture to evaluate the utility of the distilled dataset. Additionally, we present results for other models (e.g., ResNet~\citep{he2016deep}) in the \appref{appendix:cross_architecture_evaluation}, which benefit from the prior knowledge in DP-generated data to enhance cross-architecture generalizability.

\paragraph{Baselines.} We compare our proposed method, \lib, with state-of-the-art DP-DD algorithms and distilled datasets directly generated from DP data generators. For comprehensive evaluation, all experiments are conducted using three different random seeds, and we report both the mean and variance of the results.
\begin{itemize}[leftmargin=12pt,nosep]
  \item DP-DD Methods: (i) Gradient matching based method: PSG~\citep{chen2022private}; (ii) Distribution matching based method: NDPDC~\citep{zheng2024differentially}.
  \item DP Data Generator Methods: (i) Input level method: DP-MERF~\citep{harder2021dp}; (ii) Model-level method: PrivImage~\citep{li2024privimage}; (iii) Output level method: PE~\citep{lin2023differentially}.
\end{itemize}
To illustrate the gap between DP-DD and standard DD algorithms, we select DM~\citep{zhao2023dataset} and NCFM~\citep{wang2025ncfm} for comparison, neither of which has DP privacy guarantees.

\paragraph{Experimental Setup.}
Due to space constraints, we defer our implementation details to \appref{appendix:experimental_setup}.

\begin{table*}[t]
  \centering
  \begin{minipage}[t]{0.48\textwidth}
    \centering
    \textbf{(a)}

    \small
    \setlength\tabcolsep{5pt}
    \begin{tabular}{cccc}
      \toprule
      DP-Init & DP-FM & DP-EG & Test Accuracy (\%) \\
      \midrule
      \CheckmarkBold &  &  &\(48.7_{\pm 0.4}\)\\
      &\CheckmarkBold  &  &\(53.2_{\pm 0.2}\)\\
      \CheckmarkBold &  \CheckmarkBold & & \(60.8_{\pm 0.3}\)\\
      \CheckmarkBold &  \CheckmarkBold & \CheckmarkBold & \(\bf{65.5_{\pm0.4}}\)\\
      \bottomrule
    \end{tabular}
    \label{tab:ablation_components}
  \end{minipage}\hfill
  \begin{minipage}[t]{0.48\textwidth}
    \centering
    \textbf{(b)}

    \small
    \setlength\tabcolsep{6pt}
    \begin{tabular}{l|cc}
      \toprule
      \multirow{2}{*}{DP Generator} & \multicolumn{2}{c}{Test Accuracy (\%)} \\
      \cmidrule{2-3}
      & CIFAR-10 & CelebA \\
      \midrule
      DP-MERF & \(56.4_{\pm 0.3}\) & \(82.4_{\pm 0.2}\) \\
      PE & \(\bf{65.5_{\pm 0.4}}\) & \(83.1_{\pm 0.3}\) \\
      PrivImage & \(64.7_{\pm 0.4}\) & \(\bf{85.7_{\pm 0.2}}\) \\
      \bottomrule
    \end{tabular}
    \label{tab:dp_generator_selection}
  \end{minipage}
  \caption{\textbf{Ablation studies.} (a) Effectiveness of different components in \lib (CIFAR-10, \IPC=50, \(\epsilon=10\)). Init, FM and EG indicate Initialization, Feature Matching and Expert Guidance, respectively. \CheckmarkBold indicates utilized component. (b) Comparison of different DP generators used in \lib with \IPC=50.}
  \label{tab:ablation_components_and_dp_generator_selection}
\end{table*}
\subsection{Comparing Utility of Distilled Datasets}
To comprehensively evaluate the effectiveness of our proposed method, we conduct experiments to compare the utility of our distilled datasets against state-of-the-art approaches. We measure utility by the test accuracy of models trained on these distilled datasets. Table~\ref{tab:res_lowresolution} presents the results across CIFAR-10, CIFAR-100, and CelebA datasets with varying privacy budgets (\(\epsilon\)=1 and \(\epsilon\)=10) and \IPC settings. Higher test accuracy indicates better utility preservation while maintaining DP guarantees.

Our method, \lib, consistently achieves superior results under identical privacy budgets and \IPC settings. For instance, on CIFAR-10 with $\epsilon$=10 and \IPC=50, \lib attains 65.5\% accuracy, markedly surpassing existing DP-DD approaches. Notably, directly generating DP data at the \IPC scale yields lower utility than DP-DD methods, as it does not effectively distill information from private data. \lib further bridges the performance gap to standard DD, significantly enhancing distilled data usability under DP guarantees. However, we note that on CIFAR-100, both DP-DD methods exhibit a significant performance gap compared to standard DD. We attribute this to the fact that CIFAR-100 has fewer samples per class. As a result, as illustrated in \lemref{lem:subsampling_theorem}, achieving DP guarantees requires injecting more noise, which reduces the utility of the distilled dataset.

To assess the generalizability of distilled datasets, we conducted extensive cross-architecture evaluations, with detailed results presented in \appref{appendix:cross_architecture_evaluation}. \lib demonstrates superior generalizability compared to other DP-DD methods when evaluated on unseen architectures.

\subsection{Comparing Privacy of Distilled Datasets Through Membership Inference Attacks}
To evaluate privacy protection, we analyze our method's resistance to MIAs by employing the LiRA~\citep{carlini2022membership,zhao2025doesgit, wang2024garrison}, strictly following their implementation. \tabref{tab:defense_mia} illustrates the utility-privacy trade-off, where utility is measured by Test Accuracy and privacy by TPR@0.1\%FPR. The results demonstrate that \lib achieves comparable MIA resistance to other methods while offering superior data utility under the same privacy budget. This validates the correctness of \lib's overall privacy budget, derived via GDP composition (\lemref{lem:gdp_composition}) and GDP-to-DP conversion (\lemref{lem:conversion}).

\subsection{Ablation Study}

\paragraph{The Impact of Individual Components in \libn.}
To understand the contribution of each component in our method, we conduct an ablation study by selectively enabling different components. \tabref{tab:ablation_components_and_dp_generator_selection} presents the test accuracy on CIFAR-10 under different configurations. We analyze how each component (DP-Init, DP-FM, and DP-EG) contributes to the overall performance, demonstrating the necessity of our multi-stage approach for achieving optimal results.

\paragraph{The Impact of DP Generators.}
The DP-generated data required by \lib can be produced by various DP generators. We investigate their impact on the final distilled dataset utility. \tabref{tab:ablation_components_and_dp_generator_selection} compares the test accuracies on CIFAR-10 and CelebA when \lib utilizes DP-generated data from DP-MERF, PE, and PrivImage, all operating under the same privacy budget. The choice of DP generator can be dataset-dependent. For example, PE, as a training-free method, may exhibit reduced generation quality if there is a significant distributional divergence between the private and public data. This analysis helps identify suitable DP generators for maximizing utility while maintaining privacy guarantees.

\begin{table}[h]
  \centering
  \small
  \setlength\tabcolsep{4pt}
  \begin{tabular}{c|c|ccc|c}
    \toprule
    $\epsilon_{\text{total}}$ & $\mu_{\text{total}}^{\dagger}$ & $\mu_\text{G}$ & $\mu_\text{F}$ & $\mu_\text{E}$ & Test Accuracy (\%) \\
    \midrule
    \multirow{4}{*}{10} & 2.00 & 1.21 & 1.17 & 1.17 & \(57.2_{\pm 0.2}\) \\
    & 2.00 & 0.27 & 1.07 & 1.66 & \(59.7_{\pm 0.4}\) \\
    & 2.00 & 0.27 & 1.75 & 0.92 & \(61.3_{\pm 0.5}\) \\
    & \bf{2.00} & \bf{0.27} & \bf{1.48} & \bf{1.30} & \(\bf{65.5_{\pm 0.3}}\) \\
    \midrule
    20 & \bf{3.44} & \bf{0.50} & \bf{2.50} & \bf{2.31} & \(\bf{68.7_{\pm 0.4}}\) \\
    \bottomrule
  \end{tabular}
  \caption{\textbf{Ablation study on different privacy budget allocations.} The \textbf{bold rows} are derived from our strategy. Results are for CIFAR-10 with \IPC=50.
  $^{\dagger}\mu_{\text{total}}$ is derived from the target $\epsilon_{\text{total}}$ and $\delta_{\text{total}}$.}
  \label{tab:privacy_allocation}
\end{table}

\paragraph{The Impact of Privacy Budget Allocation Strategy.}
The allocation of privacy budgets across different components of \lib is a crucial design choice. \tabref{tab:privacy_allocation} presents an analysis of various privacy budget allocations on CIFAR-10 with \IPC=50. The results indicate that our strategy can significantly improve the utility of the distilled dataset. Note that DP data generation can operate with a smaller privacy budget due to the capabilities of pre-trained generators. For training the expert model, allocating a large privacy budget would introduce excessive noise in the feature matching stage, while a small budget would lead to an inaccurate class distribution for guidance.

\subsection{Visualization of Distilled Dataset}
To provide qualitative insights into the performance of our method, we visualize the distilled datasets for CIFAR-10, CIFAR-100 and CelebA, generated by different DP-DD methods. In \appref{appendix:visualization_of_distilled_dataset}, we present representative samples from PSG, NDPDC, and our method for these datasets. We observe that \lib generates data with much higher realism compared to the other methods, which enhances the cross-architecture generalization~\citep{sun2024diversity}.

\section{Conclusion}
This paper introduces \libn, a novel framework for differentially private dataset distillation (DP-DD) that overcomes the utility and realism challenges faced by existing DP-DD methods. By strategically leveraging DP-generated data for initialization, training feature extractors, and employing an expert model for guided distillation, we observe that \lib significantly outperforms state-of-the-art methods. Extensive experiments demonstrate \lib's superior performance in terms of dataset utility and robustness against privacy attacks, positioning it as a new, effective paradigm for trustworthy dataset distillation on sensitive data.

\clearpage
\section*{Acknowledgments}
This work was partly supported by the National Key Research and Development Program of China under No. 2024YFB3908400, NSFC under No. 62402418, 62402148, the Zhejiang Province's 2025 "Leading Goose + X" Science and Technology Plan under grant No.2025C02034, the Key R\&D Program of Ningbo under No. 2024Z115, and the China Postdoctoral Science Foundation under No. 2024M762829.
\bibliography{resources/aaai2026}

\clearpage
\appendix

\section{Notation}
\begin{table}[h]
    \centering
    \label{tab:notations-summary}
    \setlength{\tabcolsep}{1pt}
    \resizebox{\linewidth}{!}{
        \begin{tabular}{cp{6cm}}
            \toprule
            \textbf{Notation}   & \textbf{Description} \\
            \midrule
            $\mathcal{T}$, $\hat{\mathcal{T}}$ & the private dataset and its synthetic version generated under DP guarantees (DP-generated data).\\
            $\mathcal{S}$ & the distilled dataset. \\
            $\{\phi_{\text{F}}^{n}\}_{n=1}^{N}$, $\phi_{\text{E}}$ & the feature extractors and the expert model. \\
            $\sigma_{\text{G}}$, $\sigma_{\text{F}}$, $\sigma_{\text{E}}$ & the noise multiplier for DP data generation, feature matching, and expert model training, respectively. \\
            $\mu_{\text{G}}$, $\mu_{\text{F}}$, $\mu_{\text{E}}$ & the GDP parameter for DP data generation, feature matching, and expert model training, respectively. \\
            \IPC & the number of images per class in the distilled dataset, reflecting the degree of compression. \\
            \bottomrule
        \end{tabular}
    }
    \captionof{table}{Summary of notations.}
\end{table}

\section{Definitions}
\label{definition}
\begin{definition}[Differential Privacy~\citep{dwork2006calibrating}]
  A randomized algorithm \(\mathcal{A}\) is \((\epsilon, \delta)\)-DP if for any datasets \(D, D'\) that differ in one element, and for all \(\cO \subseteq \mathcal{R}\),
  \begin{equation}
    \Pr[\mathcal{A}(D) \in \cO] \leq \mathrm{e}^{\epsilon} \cdot \Pr[\mathcal{A}(D') \in \cO] + \delta
  \end{equation}
  where \(\Pr[\mathcal{A}(D) \in \cO]\) is the probability that \(\mathcal{A}(D)\) outputs a result in \(S\).
\end{definition}

\begin{definition}[$f$-DP and $\mu$-Gaussian DP~\citep{dong2022gaussian}]
A randomized algorithm $\mathcal{A}$ is said to satisfy $f$-differential privacy ($f$-DP) if for any datasets $D$ and $D'$ differing in one element, the following holds: $T(\mathcal{A}(D), \mathcal{A}(D')) \geq f$, where $T(\cdot, \cdot)$ is the trade-off function defined above.

Specifically, $\mathcal{A}$ satisfies $\mu$-Gaussian differential privacy ($\mu$-GDP) if it is $G_\mu$-DP, where $G_\mu(x) = \Phi(\Phi^{-1}(1-x) - \mu), \quad \mu \geq 0$, and $\Phi$ denotes the cumulative distribution function (cdf) of the standard normal distribution $\mathcal{N}(0,1)$.
\end{definition}

\section{}
\label{appendix:expert_model_training}
\begin{algorithm}
  \caption{The overview of DP-GenG}
  \label{alg:workflow}
  \textbf{Input}: Private dataset $\mathcal{T}$; Public dataset $\mathcal{D}_{\text{pub}}$ used to construct the generator $\mathcal{G}$; Number of feature extractors $N$; Number of classes $n$; Number of images per class \IPC; Number of Distillation iterations $I$; Clipping bound $C$; Learning rates $\eta_\text{F}$ and $\eta_\text{E}$; Budget allocation criterion $\mathcal{A}$. \\
  \textbf{Parameters}: DP-generated dataset $\hat{\mathcal{T}}$; Noise multipliers $\sigma_\text{G}$,$\sigma_\text{F}$ and $\sigma_\text{E}$. \\
  \textbf{Output}: The distilled dataset $\mathcal{S}$ with DP guarantees.
  \begin{algorithmic}[1]
  \STATE Determine the noise multipliers $\sigma_\text{G}$,$\sigma_\text{F}$, $\sigma_\text{E}$ based on the budget allocation criterion $\mathcal{A}$
  \STATE Generate DP synthetic data $\hat{\mathcal{T}}$ using $\mathcal{G}(\mathcal{T}, \mathcal{D}_{\text{pub}}, \sigma_\text{G})$ \COMMENT{privacy parameter $\mu_\text{G}$}
  \STATE Pre-train $N$ feature extractors $\{\phi_{\text{F}}^{n}\}_{n=1}^{N}$ and an expert model $\phi_{\text{E}}$ with $\sigma_\text{E}$ in \algref{alg:gradient_sanitizing}
  \STATE Initialize $\mathcal{S}$ with ($n$\(\times\) \IPC) samples of $\hat{\cT}$ via feature-level k-means clustering: $\text{Sample}(\hat{\mathcal{T}}, \text{IPC})$
  
  \FOR{$i=1,2,\ldots,I$}
      \STATE Randomly select a feature extractor $\phi_{\text{F}}$
      \FOR{$j=1,2,\ldots,n$} 
        \STATE Sample a mini-batch $\mathcal{B}^{\mathcal{T}}$ from $\mathcal{T}^j$ and a mini-batch $\mathcal{B}^{\mathcal{S}}$ from $\mathcal{S}^j$ \COMMENT{satisfy \lemref{lem:subsampling_theorem}}\\
        \STATE Calculate the feature matching loss through \eqref{eq:feature_matching_loss_redefined} \COMMENT{privacy parameter $\mu_\text{F}$}\\
        \STATE Update the $\mathcal{S}$ by $\mathcal{S} = \mathcal{S} - \eta_\text{F}\nabla_{\mathcal{S}}\mathcal{L}_{\text{F}}$
        \STATE Sample a mini-batch $\mathcal{R}^{\boldsymbol{y}^{\cS}}$ from $\hat{\mathcal{T}}^j$
        \STATE Calculate the expert guiding loss through \eqref{eq:expert_guidance_loss_redefined} \COMMENT{privacy parameter $\mu_\text{E}$}\\
        \STATE Update the $\mathcal{S}$ by $\mathcal{S} = \mathcal{S} - \eta_\text{E}\nabla_{\mathcal{S}}\mathcal{L}_{\text{E}}$
      \ENDFOR
  \ENDFOR
  \RETURN $\mathcal{S}$
  \end{algorithmic}
  \end{algorithm}

\begin{algorithm}
    \caption{DP Fine-tuning Expert Model}
    \label{alg:gradient_sanitizing}
    \begin{algorithmic}[1]
    \STATE \textbf{Input}: Private dataset $\mathcal{T}$, a model $\phi$ pretrained on DP generated data $\hat{\mathcal{T}}$, Batch size $b$, Learning rate $\eta$, Clip coefficient $C$, Gaussian variance $\sigma_{E}^{2}$, Max iterations $T_{m}$
        
    \STATE $T \gets 0$
    \WHILE{$T < T_{m}$}
        \STATE Randomly sample the image training batch $x_{1:b}$ from $\mathcal{T}$
        \STATE Calculate gradient $g_{1:b} \gets \nabla_{\phi}\mathcal{L}\left({\phi}\left(x_{1:b}\right)\right)$
        \STATE Scale gradient $g_{1:b}' \gets \min\left\{1,\frac{C}{\|g_{1:b}\|_{2}}\right\}g_{1:b}$
        \STATE Add Gaussian noise $\hat{g} \gets \frac{1}{b}\sum\limits_{i=1}^{b}g_{i}' + \frac{\sigma_{E}C}{b}e_{i}$ where $e \sim \mathcal{N}(0,1)$
        \STATE Update parameter $\phi \gets \phi - \eta\hat{g}$
        \STATE $T \gets T + 1$
    \ENDWHILE
    
    \STATE \textbf{Output}: The well-trained expert model: $\phi_\text{E}$
    \end{algorithmic}
    \end{algorithm}

\section{Target FID Score and Accuracy}
\label{appendix:target_FID_accuracy}
Our budget allocation strategy prioritizes determining the privacy budget for DP data generation $\mu_\text{G}$ and expert model training $\mu_\text{E}$, followed by calculating the privacy budget for DP feature matching $\mu_\text{F}$ according to \lemref{lem:gdp_composition}.

For the privacy budget of DP data generation ($\mu_\text{G}$), higher privacy budget leads to lower FID scores, resulting in generated data that is visually closer to the private data. Through empirical analysis, we found that once the FID score of the DP generated data reaches a certain threshold during DP data generation, further increasing the privacy budget does not significantly improve the distillation effect. Therefore, we allocate $\mu_\text{G}$ to achieve an FID score that is 1.2 times the FID score obtained with $\mu_\text{total}$. When $\epsilon=10$ and $\delta=10^{-5}$, $\epsilon_\text{total}$ is 2.0, and $\mu_\text{G}$ is typically set to 0.27.

For the privacy budget of expert model training ($\mu_\text{E}$), we need to ensure that the expert model achieves appropriate classification accuracy while maintaining sufficient privacy budget for feature matching. The expert model training process consists of two stages: initial training on DP-generated data, followed by DPSGD fine-tuning on private data. Based on our empirical analysis, $\mu_\text{E}$ is typically chosen to achieve 90\% of the accuracy that would be obtained when training with the full budget $\mu_\text{total}$ using only DPSGD on private data.

\section{Limitations}
\label{appendix:limitations}
While \lib significantly outperforms existing DP-DD methods, it does not match the performance of standard DD methods on datasets like CIFAR-100 (with 10\% drop when $\epsilon=10$ and \IPC=50), where each class has fewer samples. This necessitates injecting a larger amount of noise to maintain the same level of privacy guarantees. Additionally, \lib relies on a pre-trained generator (e.g., a pre-trained diffusion model) to produce DP-generated data, and its performance suffers when such a generator is unavailable.

\section{Privacy Analysis}
\label{appendix:privacy_analysis}

\begin{theorem}[DP-GENG Privacy Budget Allocation]
Given a total privacy budget $(\epsilon_{\mathrm{total}}, \delta_{\mathrm{total}})$, the DP-GENG algorithm consists of three main components: DP data generation (with privacy parameter $\mu_\mathrm{G}$), feature matching (with $\mu_\mathrm{F}$), and expert model training (with $\mu_\mathrm{E}$). The required noise scale $\sigma_\mathrm{F}$ for feature matching, under batch sampling probability $p$ and $T$ iterations, is:
\[
  \sigma_\mathrm{F} = \left[ \ln \left( \frac{T p^2}{\mu_{\mathrm{total}}^2 - \mu_\mathrm{G}^2 - \mu_\mathrm{E}^2} + 1 \right) \right]^{-1/2},
\]
where $\mu_{\mathrm{total}}$ is determined by $(\epsilon_{\mathrm{total}}, \delta_{\mathrm{total}})$ via
\[
  \delta = \Phi\left(-\frac{\epsilon}{\mu} + \frac{\mu}{2}\right) - \\e^{\epsilon} \Phi\left(-\frac{\epsilon}{\mu} - \frac{\mu}{2}\right),
\]
with $\Phi$ the standard normal CDF.
\end{theorem}

\begin{proof}
We prove the theorem in several steps, referencing \lemref{lem:conversion}, \lemref{lem:gdp_composition}, \lemref{lem:subsampling_theorem}.

\paragraph{Step 1: Conversion from $(\epsilon,\delta)$-DP to $\mu$-GDP.}~

By \lemref{lem:conversion}, for any $(\epsilon,\delta)$-DP mechanism, there exists a corresponding $\mu$-GDP mechanism such that
\begin{equation}
  \delta(\epsilon) = \Phi\left(-\frac{\epsilon}{\mu} + \frac{\mu}{2}\right) - e^{\epsilon} \Phi\left(-\frac{\epsilon}{\mu} - \frac{\mu}{2}\right).
\end{equation}

Given $(\epsilon_{\mathrm{total}}, \delta_{\mathrm{total}})$, we solve for $\mu_{\mathrm{total}}$.

\paragraph{Step 2: GDP Composition.}~

\vspace{1em}
By \lemref{lem:gdp_composition}, the sequential composition of $k$ mechanisms with $\mu_i$-GDP guarantees results in $\sqrt{\mu_1^2 + \cdots + \mu_k^2}$-GDP. For \lib, the three components yield
\begin{equation}
  \mu_{\mathrm{total}}^2 = \mu_\mathrm{G}^2 + \mu_\mathrm{F}^2 + \mu_\mathrm{E}^2.
\end{equation}
Thus,
\begin{equation}
  \mu_\mathrm{F}^2 = \mu_{\mathrm{total}}^2 - \mu_\mathrm{G}^2 - \mu_\mathrm{E}^2.
\end{equation}

\paragraph{Step 3: Parallel Composition for Disjoint Subsets.}~

\vspace{1em}
\begin{lemma}[Parallel Composition for GDP~\cite{smith2021making}]
\label{lem:parallel_composition_gdp}
Let $\{\mathcal{M}_i\}_{i=1}^k$ be a sequence of $k$ mechanisms, each satisfying $\mu_i$-GDP, and let $\{D_i\}_{i=1}^k$ be disjoint subsets of the dataset $\mathcal{D}$. The joint mechanism defined as the sequence of $\mathcal{M}_i(D \cap D_i)$ (possibly conditioned on the outputs of previous mechanisms) is $\max\{\mu_1, \mu_2, \ldots, \mu_k\}$-GDP.
\end{lemma}
In \algref{alg:workflow}, the private data is accessed through a nested loop structure: the outer loop iterates $T$ times, and the inner loop iterates over $n$ classes. In each inner loop, a subsample is drawn from the private data corresponding to a specific class, i.e., the dataset is partitioned into $n$ disjoint subsets $\{\cT_j\}_{j=1}^n$ according to class labels.

By \lemref{lem:parallel_composition_gdp}, if each mechanism in the inner loop operates on a disjoint subset $\cT_j$ and satisfies $\mu_j$-GDP, then the joint mechanism over all classes in a single outer iteration is $\max\{\mu_1, \ldots, \mu_n\}$-GDP. In our setting, since the same mechanism and sampling rate are applied to each class, all $\mu_j$ are equal, so the overall privacy loss per outer iteration is simply $\mu_\mathrm{F}$-GDP.

Therefore, the total privacy loss over all $T$ outer iterations is determined by the sequential composition of $T$ mechanisms, each with $\mu_\mathrm{F}$-GDP.  Thus, it suffices to account for the $T$ outer queries to the private data, as the inner loop over classes does not increase the privacy loss beyond that of a single class due to the disjointness property.

\paragraph{Step 4: Subsampling Theorem for GDP.}~

\begin{lemma}[Subsampling Theorem for GDP~\citep{bu2020deep}]
  \label{lem:subsampling_theorem}
The Poisson subsampled algorithm with probability \(p\) as number of iterations \(T\rightarrow\inf\), \(\cM\circ\text{Sample}_p\) satisfies
\begin{equation*}
  f=(pG_{1/\sigma}+(1-p)\mathrm{Id})^{\otimes T}\rightarrow G_{\mu},
\end{equation*}
where \(\mu=p\sqrt{T(\mathrm{e}^{1/\sigma^2}-1)}\). For batch size \(B\), \(\cM\circ\text{Sample}_p\) satisfies \(\frac{B}{n}\sqrt{T(\mathrm{e}^{1/\sigma^2}-1)}\)-GDP.
\end{lemma}

By \lemref{lem:subsampling_theorem}, for Poisson subsampling with probability $p$ over $T$ iterations, the resulting mechanism is $\mu_\mathrm{F}$-GDP with
\begin{equation}
  \mu_\mathrm{F} = p \sqrt{T (e^{1/\sigma_\mathrm{F}^2} - 1)}.
\end{equation}

\paragraph{Step 5: Solving for $\sigma_\mathrm{F}$.}~

Combining the above, we have
\begin{equation}
  \mu_\mathrm{F}^2 = p^2 T (e^{1/\sigma_\mathrm{F}^2} - 1),
\end{equation}
which gives
\begin{equation}
  \sigma_\mathrm{F} = \left[ \ln \left( \frac{p^2 T}{\mu_\mathrm{F}^2} + 1 \right) \right]^{-1/2}.
\end{equation}

Substituting $\mu_\mathrm{F}^2 = \mu_{\mathrm{total}}^2 - \mu_\mathrm{G}^2 - \mu_\mathrm{E}^2$ yields the claimed formula.
\end{proof}

\section{Experimental Setup}
\label{appendix:experimental_setup}
For CIFAR-10 and CIFAR-100, we use PE as the DP data generator, while for CelebA, we use PrivImage. The scale of DP-generated data is set to 50,000. For the initialization of the distilled dataset, we adopt the upsampling reparameterization technique with a factor parameter of 2. 

We train 20 feature extractors with different random initializations for 30 epochs on the DP-generated data $\hat{\mathcal{T}}$. For the expert model, we first train for 100 epochs on $\hat{\mathcal{T}}$, then fine-tune for 20 epochs using DPSGD on the private dataset $\mathcal{T}$. We use a batch size of 128 for the private dataset (50 for CIFAR-100). Feature matching is performed for 2,000 iterations with a learning rate of 5e-3 and a batch size of 128 (50 for CIFAR-100).

To evaluate the utility of the distilled dataset, we use the SGD optimizer with a learning rate of 0.01, momentum of 0.9, and weight decay of 0.0005 to train models. The DP generator is trained on 4 RTX A6000 GPUs, and feature matching is conducted on a single RTX A6000 GPU.
\section{Cross-Architecture Evaluation}
\label{appendix:cross_architecture_evaluation}
We evaluated the performance of DP-DD methods across different architectures by distilling datasets using ConvNet-3 and training models on ResNet-18 and DenseNet-121. The results are presented in Table~\ref{tab:cross_architecture}.
\begin{table}[!htbp]
    \centering
    \small
    \setlength{\tabcolsep}{6pt}
    \renewcommand{\arraystretch}{1.15}
    \begin{tabular}{l c ccc}
        \toprule
        \textbf{Method} & \textbf{IPC} & \textbf{ConvNet-3} & \textbf{ResNet-18} & \textbf{DenseNet-121} \\
        \midrule
        \multirow{2}{*}{PSG} 
            & 10 & 40.3{\scriptsize$\pm$0.4} & 36.5{\scriptsize$\pm$0.4} & 34.9{\scriptsize$\pm$0.8} \\
            & 50 & 47.2{\scriptsize$\pm$0.6} & 43.2{\scriptsize$\pm$0.4} & 40.9{\scriptsize$\pm$0.3} \\
        \midrule
        \multirow{2}{*}{NDPDC} 
            & 10 & 46.8{\scriptsize$\pm$0.6} & 42.3{\scriptsize$\pm$0.5} & 39.0{\scriptsize$\pm$0.1} \\
            & 50 & 53.9{\scriptsize$\pm$0.2} & 48.6{\scriptsize$\pm$0.6} & 45.5{\scriptsize$\pm$0.4} \\
        \midrule
        \multirow{2}{*}{\lib} 
            & 10 & \textbf{59.0}{\scriptsize$\pm$0.3} & \textbf{52.8}{\scriptsize$\pm$0.4} & \textbf{49.5}{\scriptsize$\pm$0.8} \\
            & 50 & \textbf{65.5}{\scriptsize$\pm$0.4} & \textbf{58.2}{\scriptsize$\pm$0.2} & \textbf{56.7}{\scriptsize$\pm$0.3} \\
        \bottomrule
    \end{tabular}
    \vspace{1em}
    \caption{Cross-architecture evaluation of different DP-DD methods on CIFAR-10 with $\epsilon=10$. Test accuracy (\%) is reported for ConvNet-3, ResNet-18, and DenseNet-121 on two IPC settings (10 and 50). Best results are in bold.}
    \label{tab:cross_architecture}
\end{table}

\section{Visualization of Distilled Dataset}
\label{appendix:visualization_of_distilled_dataset}
\begin{figure}[H]
    \centering
    \includegraphics[width=0.45\textwidth]{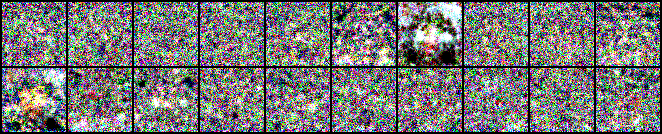}
    \caption{PSG, CelebA, $\IPC=10$.}
  \end{figure}
  \begin{figure}[H]
    \centering
    \includegraphics[width=0.45\textwidth]{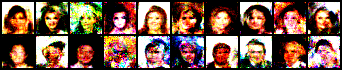}
    \caption{NDPDC, CelebA, $\IPC=10$.}
  \end{figure}
  \vspace{-5pt}
  \begin{figure}[H]
    \centering
    \includegraphics[width=0.45\textwidth]{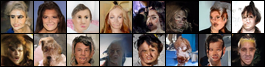}
    \caption{\lib, CelebA, $\IPC=10$.}
  \end{figure}
  \vspace{-5pt}
  \begin{figure}[H]
    \centering
    \begin{subfigure}{0.26\textwidth}
      \centering
      \includegraphics[width=\textwidth]{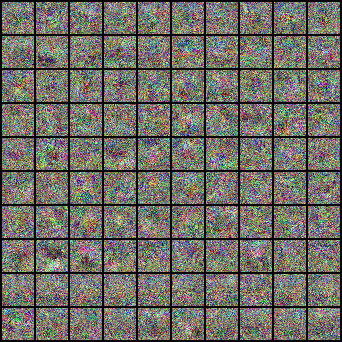}
      \caption{PSG}
      \label{fig:psg}
    \end{subfigure}
    \hfill
    \begin{subfigure}{0.26\textwidth}
      \centering
      \includegraphics[width=\textwidth]{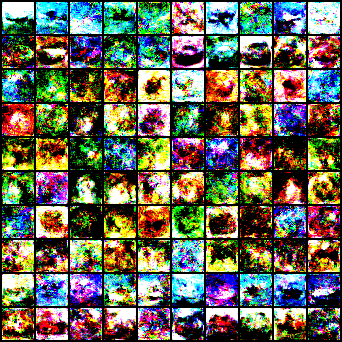}
      \caption{NDPDC}
      \label{fig:ndpdc}
    \end{subfigure}
    \hfill
    \begin{subfigure}{0.26\textwidth}
      \centering
      \includegraphics[width=\textwidth]{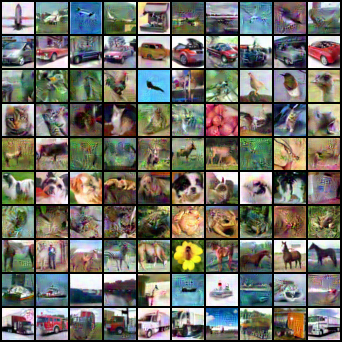}
      \caption{\lib (Ours)}
      \label{fig:ours}
    \end{subfigure}
    \caption{Different DP-DD methods on CIFAR-10 with $\IPC=50$.}
    \label{fig:distilled_datasets}
  \end{figure}

  \begin{figure}[H]
    \centering
    \begin{subfigure}{0.45\textwidth}
      \centering
      \includegraphics[width=\textwidth]{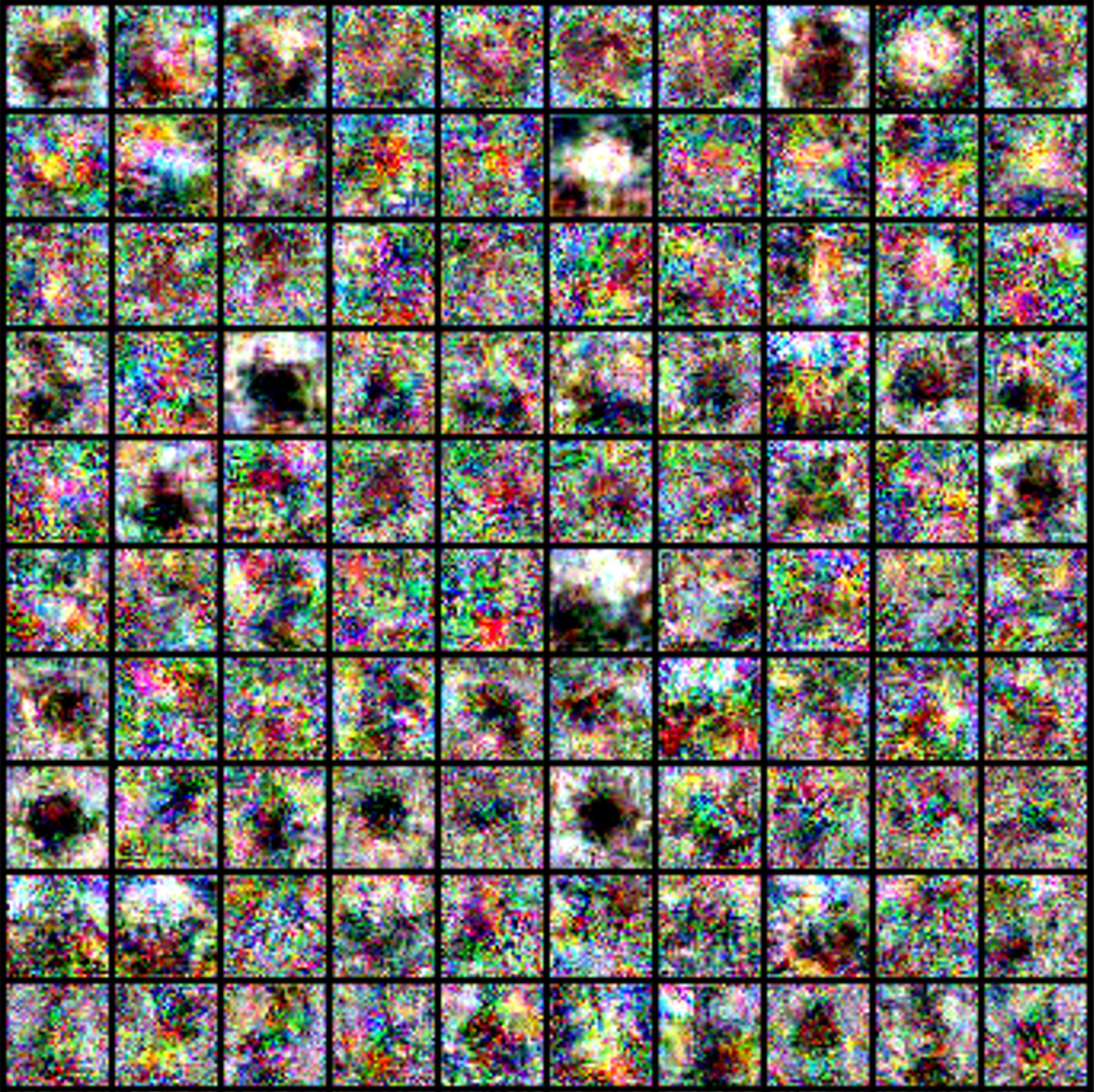}
      \caption{PSG}
      \label{fig:psg_cifar100}
    \end{subfigure}
    \hfill
    \begin{subfigure}{0.45\textwidth}
      \centering
      \includegraphics[width=\textwidth]{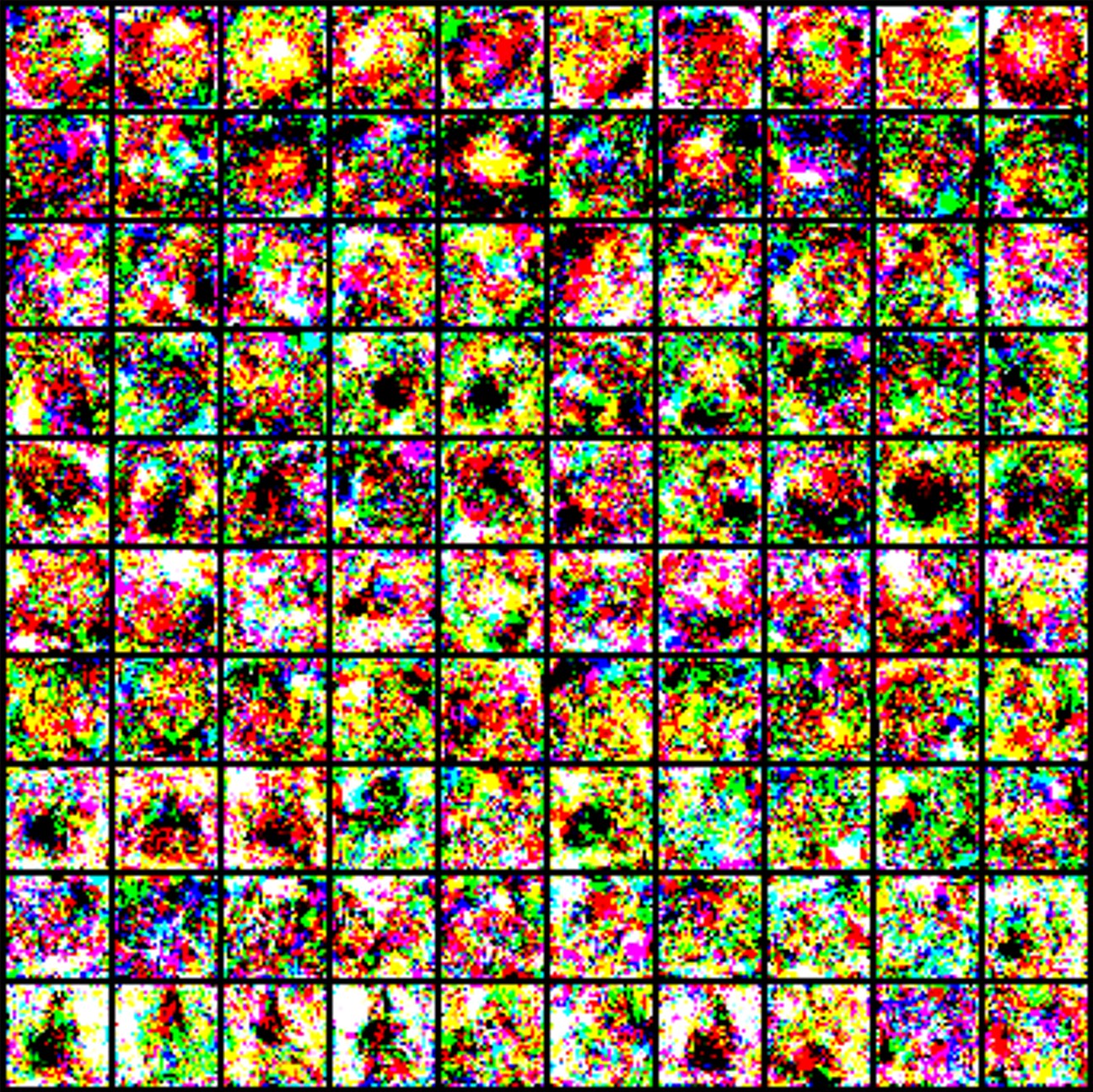}
      \caption{NDPDC}
      \label{fig:ndpdc_cifar100}
    \end{subfigure}
    \caption{The visualization of the distilled dataset generated by different methods on CIFAR-100.}
    \label{fig:cifar100_comparison}
  \end{figure}

  \section{Analysis of the Shift in Feature Representations of Distilled Examples During the Distillation Process}
\label{appendix:class_distribution_shift}
We found that during the feature matching distillation process, the distilled samples exhibit a progressive decrease in confidence for their own class labels. Consequently, this leads to a decline in their overall discriminability, making them less distinct from other classes.

  \begin{figure}[H]
    \centering
    \begin{subfigure}{0.45\textwidth}
      \centering
      \includegraphics[width=\textwidth]{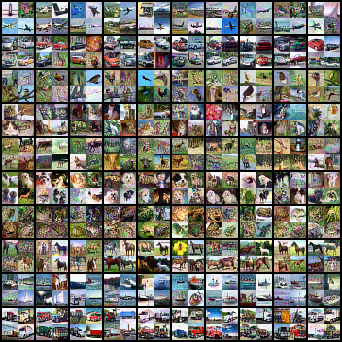}
      \caption{CIFAR-10}
      \label{fig:dpgeng_factor}
    \end{subfigure}
    \hfill
    \begin{subfigure}{0.45\textwidth}
      \centering
      \includegraphics[width=\textwidth]{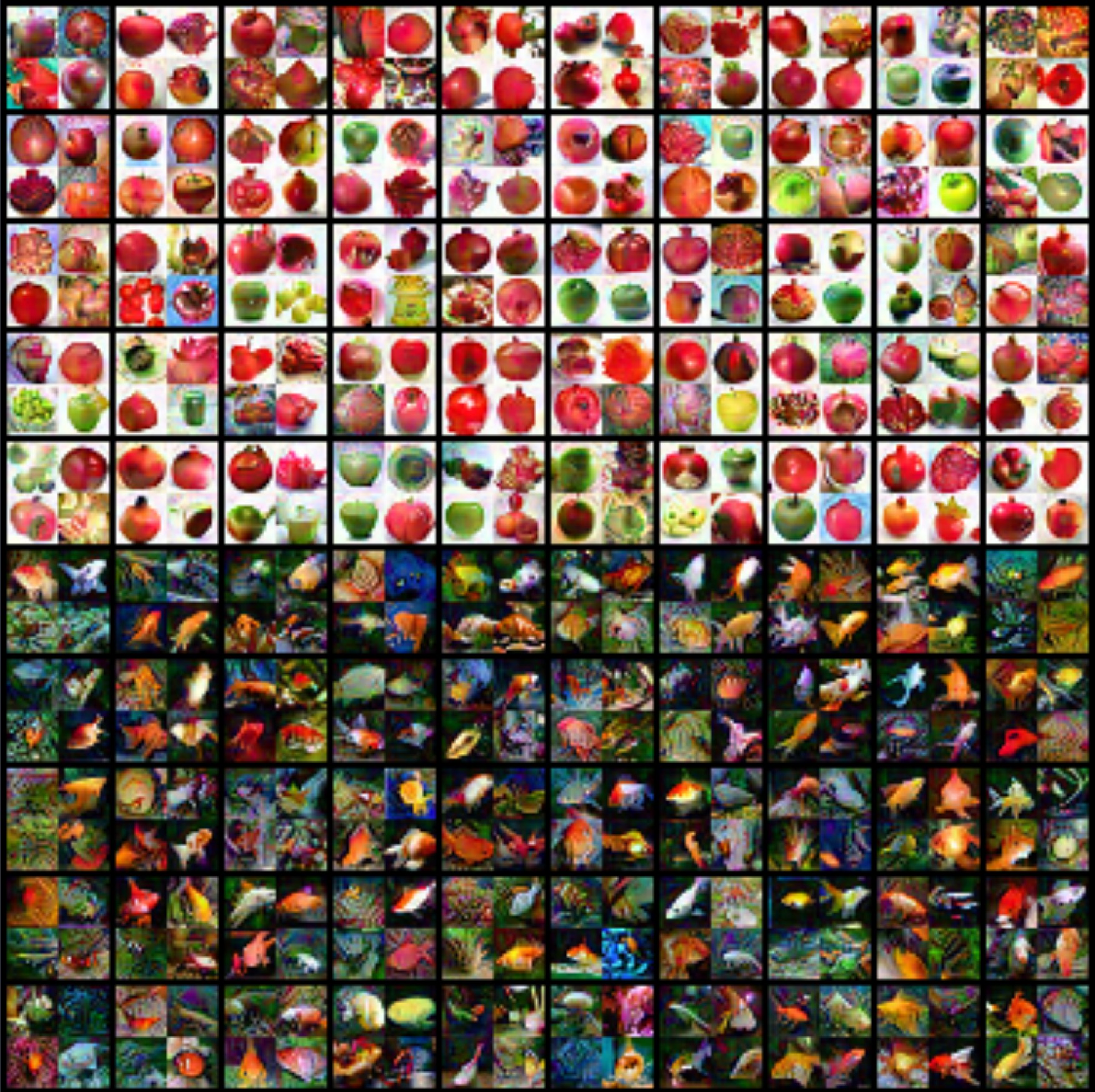}
      \caption{CIFAR-100}
      \label{fig:dpgeng_cifar100}
    \end{subfigure}
    \caption{The visualization of the distilled dataset generated by \lib with upsampling reparameterization.}
    \label{fig:dpgeng_comparison}
  \end{figure}

\begin{figure*}
    \centering
    \includegraphics[width=0.8\textwidth]{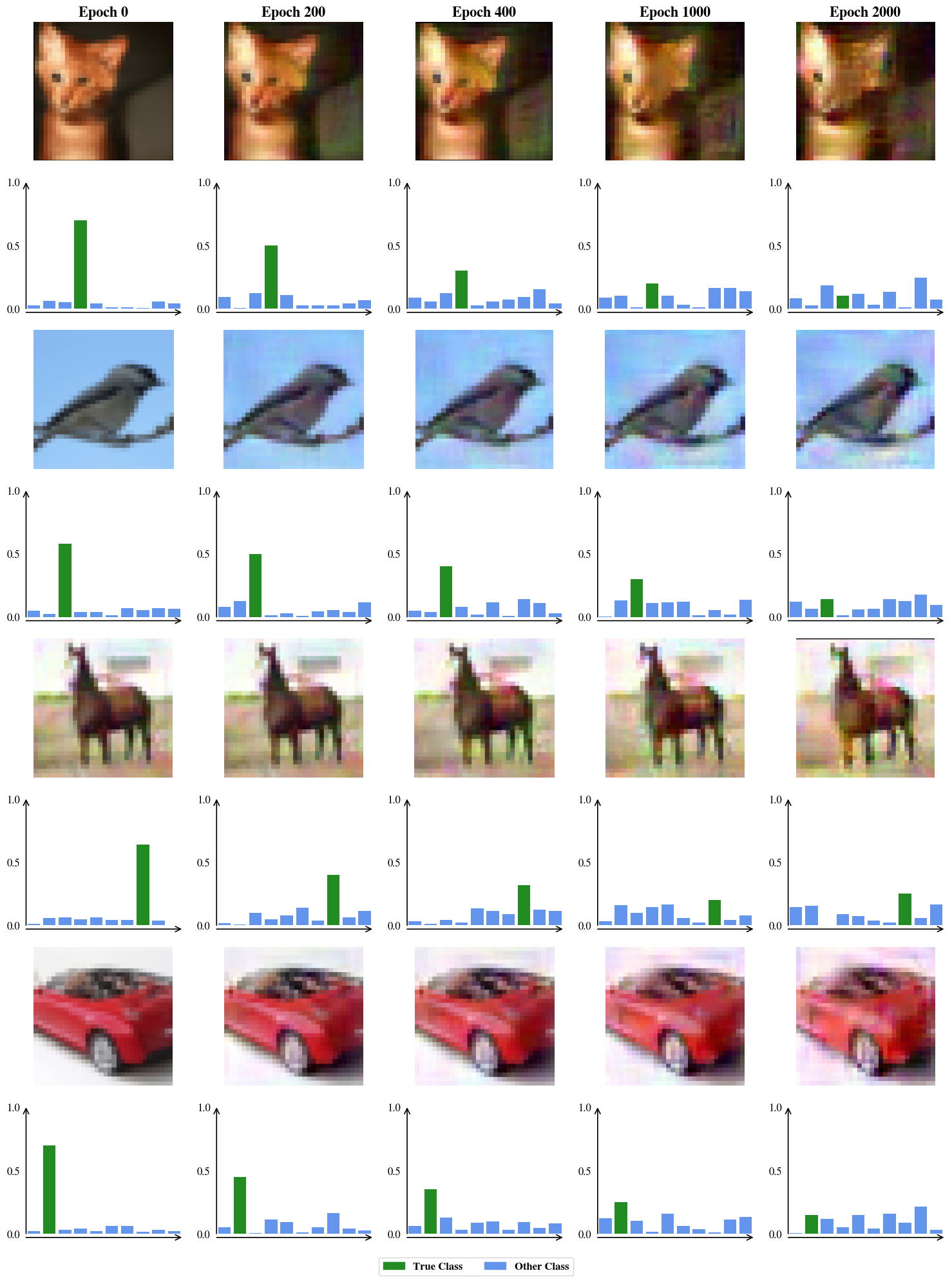}
    \captionof{figure}{The analysis of the shift in feature representations of distilled examples during the distillation process. This figure illustrates how the features of distilled examples deviate from their intended class due to the DP noise introduced during the distillation process. This observation motivates the design of our expert guidance mechanism.} 
    \label{fig:appendix_dist_shift}
\end{figure*}

\newpage

\end{document}